\documentclass[USenglish]{lipics-v2016}

\usepackage[roman]{complexity}
\usepackage{graphicx}
\usepackage{tikz}
\usepackage{tkz-graph}
\usepackage[ruled,vlined]{algorithm2e}
\usepackage{xifthen}

\newtheorem{proposition}[theorem]{\bf Proposition}
\newtheorem{openproblem}{\bf Open problem}

\newcommand{\N}{\mathbb{N}}
\newcommand{\kSAT}[1]{#1\text{-}\SAT}
\newcommand{\Pad}[1]{\mathsf{Pad}_{#1}}
\newcommand{\ETH}{\mathsf{ETH}}
\newcommand{\SETH}{\mathsf{SETH}}
\newcommand{\enum}[1]{\Pi_{#1}}
\newcommand{\Enum}[1]{\mathrm{Enum\cdot#1}}
\newcommand{\EnumP}{\mathrm{ EnumP}}
\newcommand{\En}{\mathsf{En}}
\newcommand{\Ex}{\mathsf{Ex}}

\newcommand{\OutputP}{\mathrm{ OutputP}}
\newcommand{\IncP}{\mathrm{IncP}}
\newcommand{\UIncP}{\mathrm{UsualIncP}}
\newcommand{\DelayP}{\mathrm{DelayP}}

\newcommand{\pr}{\mathbb{P}}
\newcommand{\timev}[3][]{\ifthenelse{\isempty{#2}}{t[#1]}{\ifthenelse{\isempty{#1}}{t_{#2}^{#3}}{t_{#2}^{#3}[#1]}}}

\renewcommand{\th}{\text{th}}

\title{On the Complexity of Enumeration}
\author[1]{Florent Capelli}
\author[2]{Yann Strozecki}
\affil[1]{Birkbeck University, London}
\affil[2]{Université de Versailles Saint-Quentin-en-Yvelines, DAVID laboratory}
\Copyright{Florent Capelli and Yann Strozecki}
\subjclass{F.2.2 Nonnumerical Algorithms and Problems}
\keywords{enumeration, incremental time, polynomial delay, structural complexity, exponential time hypothesis}
\SetKwInput{KwData}{Input}

\begin{document}

\maketitle

\begin{abstract}
We investigate the relationship between several enumeration complexity classes and focus in particular on problems having enumeration algorithms with incremental and polynomial delay ($\IncP$ and $\DelayP$ respectively). We show that, for some algorithms, we can turn an average delay into a worst case delay without increasing the space complexity, suggesting that $\IncP_1 = \DelayP$ even with polynomially bounded space. We use the Exponential Time Hypothesis to exhibit a strict hierarchy inside $\IncP$ which gives the first separation of $\DelayP$ and $\IncP$. Finally we relate the uniform generation of solutions to probabilistic enumeration algorithms with polynomial delay and polynomial space. 
\end{abstract}

\section{Introduction}

An enumeration problem is the task of listing a set of elements, usually corresponding to the solutions of a search problem, such as enumerating the spanning trees of a given graph or the satisfying assignments of a given formula.
One way of measuring the complexity of an enumeration algorithm is to evaluate how the \emph{total time} needed to compute all solutions relates with the size of the input and with the size of the output, as the number of solutions may be exponential in the size of the input. Therefore, a problem 
is considered tractable and said to be \emph{output polynomial} when it can be solved in polynomial time in the size of the \emph{input and the output}. This measure is relevant when one wants to generate and store all elements of a set, for instance to constitute a library of interesting objects, as it is often done in biology or chemistry~\cite{barth2015efficient}. 

Another application is to use enumeration algorithms to compute optimal solutions by generating them all or to compute statistics on the set of solutions such as evaluating its size. If this set is too large, it can be interesting to generate only a fraction of it. Hence, a good algorithm for this purpose should guarantee that it will find as many solutions as we need in a reasonable amount of time. In this case, \emph{polynomial incremental time} algorithms are more suitable: an algorithm is in  polynomial incremental time if the time needed to enumerate the first $k$ solutions is polynomial in $k$ and in the size of the input. Such algorithms naturally appear when the enumeration task is of the following form: given a set of elements and a polynomial time function acting on tuples of elements, produce the closure of the set by the function. One can generate such closure by iteratively applying the function until no new elements are found. As the set grows bigger, finding new elements becomes harder. For instance, the best algorithm to generate all circuits of a matroid uses some closure property of the circuits~\cite{khachiyan2005complexity} and is thus in polynomial incremental time. The fundamental problem of generating the minimal transversals of a hypergraph can also be solved in subexponential incremental time~\cite{fredman1996complexity} and some of its restrictions in polynomial incremental time~\cite{eiter2003new}.

Polynomial incremental time algorithms are not always satisfactory as the delay between the last solutions may be exponentially large. In some cases, the user may be more interested in having a regular stream of solutions. The most efficient enumeration algorithms guarantee a \emph{delay} between consecutive solutions that is bounded by a polynomial in the input. \emph{Polynomial delay} algorithms produce solutions regularly and generate the set of solutions in time linear in the size of the output, which can still be overall exponential. There exists two main methods giving polynomial delay algorithms, namely the \emph{backtrack search} and the \emph{reverse search}~\cite{mary2013enumeration}. These methods have been used to give polynomial delay algorithms for enumerating the cycles of a graph~\cite{read1975bounds}, the satisfying assignments of variants of $\SAT$~\cite{creignou1997generating}, the spanning trees and connected induced subgraphs of a graph~\cite{avis1996reverse} etc. These methods are particularly efficient as they only need a \emph{polynomial space}, which is required in practice. Another approach used to enumerate elements of a set while using only polynomial space, is to design and use random generators of solutions, a very active area of research~\cite{duchon2004boltzmann}. Following Goldberg~\cite{Goldberg91}, we also give precise connections between the existence of efficient random generators and efficient randomized enumeration algorithms.
 
Enumeration algorithms  have been studied for the last $40$ years~\cite{read1975bounds} and the notions of incremental polynomial time and polynomial delay already appear in~\cite{johnson1988generating}. However, the structural complexity of enumeration has not been investigated much, one reason being that it seems harder to formalize than decision or counting complexity. Recent work gives a framework for studying the parametrized complexity of enumeration problems~\cite{creignou2017paradigms} and an analogue of the polynomial hierarchy for the incremental time has been introduced~\cite{creignou2017complexity}. The complexity of enumeration when the order of the output is fixed has also been studied, for instance in~\cite{CreignouOS11}. However, from the point of view of structural complexity, it makes enumeration complexity artificial and it mainly boils down to decision complexity as explained in Section $2.4$ of \cite{phdstrozecki}.

The main difficulty in the study of structural complexity of enumeration is that complete problems are known only for $\EnumP$, the equivalent of $\NP$ in enumeration, but not for the other natural classes. In this paper, we therefore focus on understanding and separating these classes by using classical hypotheses in complexity theory. Such hypotheses are needed since we ask the generated solutions to be checkable in polynomial time, a reasonable assumption which makes separation of classes much harder.
The aim of this paper is twofold. First, we would like it to be usable as a short survey giving the definition of the main enumeration complexity classes with context and open problems as well as folklore results which were scattered over several unpublished work and thesis or only implicitly stated in a proof~\cite{Goldberg91,Schmidt09,phdstrozecki,Bagan09,brault2013pertinence,mary2013enumeration}. Second, we prove several new results which connects enumeration complexity to other fields such as fined grained complexity (Exponential Time Hypothesis), total search functions ($\TFNP)$ or the count-distinct problem (HyperLogLog).

This article is organized as follows: Sec.~\ref{sec:definition} is dedicated to the definition of the complexity classes either with polynomial time checkable solutions or not. We use classical complexity hypotheses to prove separation between most classes and provide an equivalence between the separation of incremental and output polynomial time and $\TFNP \neq FP$.
In Sec.~\ref{sec:amortization}, we recall how we can simulate algorithms in \emph{linear incremental time} with \emph{polynomial delay} algorithms if we allow an exponential space. We also prove that a linear incremental time algorithm which is sufficiently regular can be turned into and polynomial delay and polynomial space algorithms, paving the way for a proof that the two classes are equal. In Sec.~\ref{sec:hier}, we prove new separation results by using the Exponential Time Hypothesis ($\ETH$). More precisely, we exhibit a strict natural hierarchy inside classes of problems having incremental polynomial time algorithms which implies a separation between polynomial delay and incremental polynomial time, the last classes not yet separated. This separation is the first in enumeration complexity to rely on $\ETH$ and we believe it can lead to new conditional lower bounds on natural enumeration problems. Finally, in Sec.~\ref{sec:uniform}, we consider enumeration problems whose solutions can be given by a polynomial time uniform random generator. We improve a result of~\cite{Goldberg91} which shows how to turn a uniform random generator
into a randomized polynomial delay algorithm with \emph{exponential space}. We also show how to get rid of the exponential space if we are willing to allow repetitions by using algorithms to approximate the size of a dynamic set~\cite{kane2010optimal}.

\section{Complexity Classes}\label{sec:definition}
Let $\Sigma$ be a finite alphabet and $\Sigma^*$ be the set of finite words built on $\Sigma$.
We assume that our alphabet is $\{0,1,\sharp\}$. 
We denote by $|x|$ the size of a word $x \in \Sigma^*$ and by $|S|$ the cardinal of a set $S$.
We recall here the definition of an enumeration problem:
\begin{definition}[Enumeration Problem]
Let $A\subseteq \Sigma^{*}\times\Sigma^{*}$ be a binary predicate, we write $A(x)$  for the set of $y$ such that $A(x,y)$ holds. The enumeration problem $\enum{A}$ is the function which associates $A(x)$ to $x$.
\end{definition}

From now on, we only consider predicates $A$ such that $A(x)$ is finite for all $x$.
This assumption could be lifted and the definitions on the complexity of enumeration adapted to the infinite case.
We chose not to do so to lighten the presentation and because infinite sets of solutions implies some artificial 
properties when studying the complexity of enumeration. However there are interesting infinite enumeration problems such as listing all primes or all words of a context-free language~\cite{florencio2015naive}.

The computational model is the random access machine model (RAM) with addition, subtraction
and multiplication as its basic arithmetic operations. We have additional output registers,
and when a special \textsc{output} instruction is executed, the content of the output registers is produced.
A RAM machine solves $\enum{A}$ if, on every input $x \in \Sigma^{*}$, it produces a sequence $ y_{1}, \dots, y_{n}$ such that $ A(x) = \left\lbrace y_{1}, \dots, y_{n} \right\rbrace $ and for all $i\neq j,\, y_{i} \neq y_{j}$.

To simplify the definitions of complexity classes, we ask the RAM machine to stop
immediately after the last \textsc{output} instruction is executed. The cost of every instruction is assumed
to be in $O(1)$ except the arithmetic instructions which are of cost linear in the size of their inputs.
The space used by the machine at a given step is the sum of the number of bits required to store the integers in its registers. 

We denote by $T(M,x,i)$ the sum of the costs of the instructions executed before the $i^{\text{th}}$ \textsc{output} instruction. 
Usually the machine $M$ will be clear from the context and we will write $T(x,i)$ instead of $T(M,x,i)$.

\subparagraph{The class EnumP.}

We can naturally define complexity classes of enumeration problems by restricting the predicate $A(x,y)$ used to define enumeration problems. 

\begin{definition} 
Let $\cal{C}$ be a set of binary predicates, $\Enum{\cal{C}}$ is the set of problems $\enum{A}$ such that $A \in \cal{C}$.
\end{definition}

As we have explained, we restrict to the enumeration of finite sets: we let $\mathrm{F}$ be the set of all $A$ such that, for all $x$, $A(x)$ is finite and we will often consider $\Enum{F}$ as the most general class of enumeration problems. 

We are mostly interested in the class of problems which are the enumeration of the solutions of an $\NP$ problem.
Let $\mathrm{PtPb}$ be the set of predicates $A$ such that $A(x,y)$ is decidable in \textbf{P}olynomial \textbf{t}ime and is
\textbf{P}olynomially \textbf{b}alanced that is the elements of $A(x)$ are of size polynomial in $|x|$. We will denote the class $\Enum{PtPb}$ by $\EnumP$ for resemblance with $\NP$ as it is done in~\cite{phdstrozecki}.

The class $\EnumP$ has complete problems for the parsimonious reduction borrowed from counting complexity.
\begin{definition}[Parsimonious Reduction]
Let $\enum{A}$ and $\enum{B}$ be two enumeration problems. 
A parsimonious reduction from $\enum{A}$ to $\enum{B}$ is a pair of polynomial time computable functions $f, g$ such that for all $x$, $g(x)$ is a bijection between $A(x)$ and $B(f(x))$.
\end{definition}

An $\EnumP$-complete problem is defined as a problem in $\EnumP$ to which any problem in $\EnumP$ reduces.
 The problem $\enum{SAT}$, the task of listing all solutions of a $3$-CNF formula is $\EnumP$-complete, since the reduction used in the proof that SAT is $\NP$-complete~\cite{cook1971complexity} is parsimonious.
 The parsimonious reduction is enough to obtain $\EnumP$-complete problem, but is usually too strong to make some natural candidates complete problems. For instance if we consider the predicate $SAT0(\phi,x)$ which is true if and only if $x$ is a satisfying assignment of the propositional formula $\phi$ or $x$ is the all zero assignment, then $SAT0(\phi)$ is never empty and therefore many problems of $\EnumP$ cannot be reduced to $\enum{SAT0}$ by parsimonious reduction. Many other reductions have been considered~\cite{mary2013enumeration}, inspired by the many one reduction, the Turing reduction or reductions for counting problems~\cite{durand2000subtractive}. However, no complete problems are known for the complexity classes we are going to introduce with respect to any of these reductions. This emphases the need to prove separations between enumeration complexity classes since we cannot rely on reductions to understand the hardness of a problem with regard to a complexity class.

\subparagraph{The class OutputP.}

To measure the complexity of an enumeration problem, we consider the total time taken to compute all solutions. Since the number of solutions can be exponential with regard to the \emph{input}, it is more relevant to give the total time as a function of the size of the input and of the \emph{the output}. In particular, we would like it to be polynomial in the number of solutions; algorithms with this complexity are said to be in output polynomial time or sometimes in polynomial total time. We define two corresponding classes, one when the problem is in $\EnumP$ and one when it is not restricted.

\begin{definition}[Output polynomial time]
 A problem $\enum{A}\in \EnumP$ (respectively, in $\Enum{F}$) is in $\OutputP$ (resp., $\OutputP^F$) if there is a polynomial $p(x,y)$ and a machine $M$ which solves $\enum{A}$ and such that for all $x$, $T(x,|A(x)|) < p(|x|,|A(x)|)$.
\end{definition}

For instance, if we see a polynomial as a set of monomials, then classical algorithms for interpolating multivariate polynomials from their values are output polynomial~\cite{zippel1990interpolating} as they produce the polynomial in a time proportional to the number of its monomials.

\begin{proposition}\label{prop:output}
 $\OutputP = \EnumP$ if and only if $\P = \NP$.
\end{proposition}
\begin{proof}
 Assume $\OutputP = \EnumP$, thus $\enum{SAT}$ is in $\OutputP$. Then on an instance $x$, it can be solved in time bounded by $p(|x|)q(|SAT(x)|)$ where $p$ and $q$ are two polynomials. Let $c$ be the constant term of $q$, if we run the enumeration algorithm for $\enum{SAT}$ and it does not stop before a time $cp(|x|)$, we know there must be a least an element in $SAT(x)$. If it stops before a time $cp(|x|)$, it produces the set $SAT(x)$ therefore we can decide the problem $SAT$ in polynomial time.
 
 Assume now that $\P = \NP$. The problem $SAT$ is autoreducible, that is given a formula $\phi$
 and a partial assignment of its variables $a$, we can decide whether $a$ can be extended 
 to a satisfying assignment by deciding $SAT$ on another instance. Therefore we can decide in polynomial
 time if there is an extension to a partial assignment and by using the classical backtrack search or flashlight method (see for instance~\cite{mary2016efficient}) we obtain an $\OutputP$ algorithm for $\enum{SAT}$, which by completeness of $\enum{SAT}$ for $\EnumP$ yields $\EnumP = \OutputP$.
\end{proof}

The classes $\EnumP$ and $\OutputP$ may be seen as analog of $\NP$ and $\P$ for the enumeration.
Usually an enumeration problem is considered to be tractable if it is in $\OutputP$, especially if its complexity is linear in the number of solutions. The problems in $\OutputP$ are easy to solve when there are few solutions and hard otherwise. We now introduce classes of complexity inside $\OutputP$ to capture the problems which could be considered as classes of tractable problems even when the number of solutions is high.

\subparagraph{The class IncP.}

From now on, a polynomial time precomputation step is always allowed before the start of the enumeration. It makes the classes of complexity more meaningful, especially their fine grained version. It is usually used in practice to set up useful datastructures or to preprocess the instance.

Given an enumeration problem $A$, we say that a machine $M$ enumerates $A$ in {\em incremental time} $f(m)g(n)$ if on every input $x$, $M$ enumerates $m$ elements of $A(x)$ in time $f(m)g(|x|)$ for every $m \leq |A(x)|$.

\begin{definition}[Incremental polynomial time]
A problem $\enum{A} \in \EnumP$ (respectively, in $\Enum{F}$) is in $\IncP_a$ (resp. $\IncP_a^F$) if there is a machine $M$ which solves it in incremental time $cm^an^b$ for $b$ and $c$ constants. Moreover, we define $\IncP = \bigcup_{a \geq 1} \IncP_a$ and $\IncP^F = \bigcup_{a \geq 1} \IncP_a^F$.
\end{definition}

Let $A$ be a binary predicate, $\mathsf{AnotherSol_A}$ is the search problem defined as given $x$ and a set $\mathcal{S}$, find $y \in A(x) \setminus \mathcal{S}$ or answer that $\mathcal{S} \supseteq A(x)$ (see:\cite{phdstrozecki,creignou2017complexity}). The problems in $\IncP$ are the ones with a polynomial search problem: 

\begin{proposition}[Proposition $1$ of \cite{phdstrozecki}]
\label{prop:anothersolincp}
  Let $A$ be a predicate such that $\enum{A} \in \EnumP$. $\mathsf{AnotherSol_A}$ is in $\FP$ if and only if $\enum{A}$ is in $\IncP$.
\end{proposition}
\begin{proof}
  First assume that $\mathsf{AnotherSol_A}$ is in $\FP$. Given $x$, we can enumerate $A(x)$ using the following algorithm: we start with $\mathcal{S} = \emptyset$ and iteratively add solutions to $\mathcal{S}$ by running $\mathsf{AnotherSol_A}(x,\mathcal{S})$ until no new solution is found, that is, until $\mathcal{S} = A(x)$. The delay between the discovery of two new solutions is polynomial in $|\mathcal{S}|$ and $|x|$ since $\mathsf{AnotherSol_A}$ is in $\FP$. Thus, $\enum{A}$ is in $\IncP$.

  Now assume that $\enum{A}$ is in $\IncP$. That is, we have an algorithm $M$ that given $x$, output $k$ different elements of $A(x)$ in time $c|x|^ak^b$ for $a,b,c$ constants. Given $x$ and $\mathcal{S}$, we solve $\mathsf{AnotherSol_A}(x,\mathcal{S})$ in polynomial time as follows: we simulate $M$ for $c|x|^a|(1+|\mathcal{S}|)^b$ steps. If the algorithm stops before that, then we have completely generated $A(x)$. It is then sufficient to look for $y \in A(x) \setminus \mathcal{S}$ or, if no such $y$ exists, output that $\mathcal{S} \supseteq A(x)$. If the algorithm has not stopped yet, then we know that we have found $|\mathcal{S}|+1$ elements of $A(x)$. At least one of them is not in $\mathcal{S}$ and we return it.

\end{proof}

The class $\IncP$ is usually defined as the class of problems solvable by an algorithm with a delay polynomial in the number of already generated solutions and in the size of the input. This alternative definition is motivated by saturation algorithms, which generates solutions by applying some polynomial time rules to enrich the set of solutions until saturation. There are many saturation algorithms, for instance to enumerate circuits of matroids~\cite{khachiyan2005complexity} or to compute closure by set operations~\cite{mary2016efficient}.

\begin{definition}[Usual definition of incremental time.]
A problem $\enum{A} \in \EnumP$ (respectively in $\Enum{F}$) is in $\UIncP_a$  if there is a machine $M$ which solves it such that for all $x$ and for all $ 0 < t \leq |A(x)|$, $| T(x,t) - T(x,t-1)| < c t^a |x|^b$ for $b$ and $c$ constants. Moreover, we define $\UIncP = \bigcup_{a \geq 1} \UIncP_a$.
 \end{definition}

With our definition, we capture the fact that investing more time guarantees more solutions to be output, which is a bit more general at first sight than bounding the delay because the time between two solutions is not necessarily regular. We will see in Sec.~\ref{sec:amortization} that both definitions are actually equivalent but the price for regularity is to use exponential space.  

We now relate the complexity of $\IncP$ to the complexity of the class $\TFNP$ introduced in~\cite{megiddo1991total}.
A problem in $\TFNP$ is a polynomially balanced polynomial time predicate $A$ such that for all 
$x$, $A(x)$ is not empty. An algorithm solving a problem $A$ of $\TFNP$ on input $x$ outputs one element of $A(x)$.
The class $\TFNP$ can also be seen as the functional version of $\NP \cap \coNP$.

\begin{proposition}
If $\TFNP = \FP $ if and only if $\IncP = \OutputP$.
\end{proposition}

\begin{proof}
Let $A$ be in $\TFNP$ and let $q$ be a polynomial such that if $A(x,y)$ then 
$|y| \leq q(|x|)$. We define $C(x,y\sharp w)$ the predicate which is true if and only if $A(x,y)$ and $|w| \leq q(|x|)$.
Observe that the set $C(x)$ is never empty by definition of $\TFNP$.
Thanks to the padding, there are more than $2^{|w|} = 2^{q(|x|)}$ elements in $C(x)$ for each $y$ such that $A(x,y)$.
Therefore the trivial enumeration algorithm testing every solution of the form $y\sharp w$ is polynomial in the number of solutions, which proves that $\enum{C}$ is in $\OutputP$.

If $\IncP = \OutputP$, we have an incremental algorithm
for $\enum{C}$. In particular, it gives, on any instance $x$, the first solution  $y\sharp w$ in polynomial time.
This is a polynomial time algorithm to solve the $\TFNP$ problem $A$, thus $\TFNP = \FP$.

Now assume that $\TFNP = \FP$ and let $\enum{A}$ be in $\OutputP$. 
We assume w.l.o.g. that the predicate $A$ is defined over $(\{0,1\}^*)^2$ and we define the relation $D((x,S),y)$ which is true if and only if 
\begin{itemize}
 \item either $y \in A(x) \setminus S$,
 \item or $y = \sharp$ and $S \supseteq A(x)$.
\end{itemize}

We show that $D$ is in $\TFNP$. First, observe that the relation $D$ is total by construction. Now, since $\enum{A} \in \OutputP \subseteq \EnumP$, the $y$ such that $A(x,y)$ holds are of size polynomial in $|x|$ which proves that $D$ is polynomially balanced.

It remains to show that one can decide $D((x,S),y)$ in time polynomial in the size of $x$, $S$ and $y$. The algorithm is as follows: if $y \neq \sharp$, then $D((x,S),y)$ holds if and only if $y \in A(x) \setminus S$. Testing whether $y \notin S$ can obviously be done in polynomial time in the size of $y$ and $S$. Now, recall that $\enum{A} \in \EnumP$, thus we can also test whether $y \in A(x)$ holds in polynomial time. 

Now assume that $y = \sharp$. Then $D((x,S),\sharp)$ holds if and only if $S \supseteq A(x)$. By assumption, $A \in \OutputP$, thus we have an algorithm that given $x$, generates $A(x)$ in time $c|x|^a|A(x)|^b$ for constants $a,b,c$. We simulate this algorithm for at most $c|x|^a|S|^b$ steps. If the algorithm stops before the end of the simulation, then we have successfully generated $A(x)$ and it remains to check if $S \supseteq A(x)$ which can be done in polynomial time. Now, if the algorithm has not stopped after having simulating $c|x|^a|S|^b$ steps, it means that $|A(x)| > |S|$. Thus, $S \nsupseteq A(x)$ and we know that $D((x,S),\sharp)$ does not hold.


We have proved that $D \in \TFNP$. Since we have assumed that $\TFNP = \FP$ we can, given 
$(x,S)$, find $y$ such that $y \in A(x) \setminus S$ or decide there is none. In other words 
the problem $\mathsf{AnotherSol_A}$ is in $\FP$ and it implies that $\enum{A} \in \IncP$ by Proposition~\ref{prop:anothersolincp}.
\end{proof}

This is yet a new link between complexity of enumeration and another domain of computer science,
namely the complexity of total search problem. It is interesting since enumeration complexity is often understood only
by relating it to decision complexity, as in Prop.~\ref{prop:output}.
Moreover recent progress on the understanding of $\TFNP$ may help us to understand the class $\IncP$. For instance, it has been proven that reasonable assumptions such as the existence of one way functions are enough to imply $\FP \neq \TFNP$~\cite{hubacek2016journey} and thus $\IncP \neq \OutputP$.

Observe that without the requirement to be in $\EnumP$, incremental polynomial time and output polynomial time are separated unconditionally.
\begin{proposition}
 $\IncP^F \neq \OutputP^F$.
\end{proposition}
\begin{proof}
Choose any $\EXP$-complete decision problem $L$ and let $A$ be the predicate such that $A(x,y)$ holds if and only if $x = 0\sharp i$ if $x \in L$ or $1\sharp i$ if $x \notin L$ with $0 \leq i  < 2^{|x|}$.
 Therefore $\enum{A}$ is easy to solve in linear total time, but since $\EXP \neq \P$ we cannot produce 
 the first solution in polynomial time and thus $\enum{A}$ is not in incremental polynomial time. 
\end{proof}

\subparagraph{The class DelayP.}

We now define the polynomial delay which by definition is a subclass of $\IncP_1$.
In Sec.~\ref{sec:amortization} we study its relationship with $\IncP_1$, while in Sec.\ref{sec:hier} we prove its separation from $\IncP$. 

\begin{definition}[Polynomial delay]
A problem $\enum{A} \in \EnumP$ (respectively in $\Enum{F}$) is in $\DelayP$ (resp. in $\DelayP^F$) if there is a machine $M$ which solves it such that for all $x$ and for all $ 0 < t \leq |A(x)|$, $| T(x,t) - T(x,t-1)| < C|x|^c$ for $C$ and $c$ constants.
 \end{definition}
 
Observe that, by definition, $\DelayP = \UIncP_0$.

\section{Space and regularity of enumeration algorithms}\label{sec:amortization}

The main difference between $\IncP_1$ and $\DelayP$ is the regularity of the delay between two solutions. In several algorithms, for instance to generate maximal cliques~\cite{johnson1988generating}, an exponential queue is used to store results, which are then output regularly to guarantee a polynomial delay. This is in fact a general method which can be used to prove that  $\IncP_1 = \DelayP$ and, more generally, $\IncP_{a+1} = \UIncP_a$. 

\begin{proposition}\label{prop:UInca}
For every $a \in \N$,  $\IncP_{a+1} = \UIncP_{a}$.
\end{proposition}
\begin{proof}
 Let $\enum{A} \in \UIncP_a$, then there is an algorithm $I$ and constants $C$ and $c$ such that $I$ on input $x$ produces $k$ solutions in time bounded by 
 \[ 
 \begin{aligned}
 \sum_{i=0}^k C  |x|^c  i^a & = C|x|^c  (\sum_{i=0}^k i^a) \\
 & \leq C|x|^c (k+1)k^a \\
 & \leq 2C  |x|^c  k^{a+1}.
 \end{aligned}
 \]
Thus $\enum{A} \in \IncP_{a+1}$.

Now let $\enum{A} \in \IncP_{a+1}$, then there is an algorithm $I$ which on an instance of size $n$, produces $k$ solutions in time bounded by $k^{a+1}p(n)$ where $p$ is a polynomial. 

We construct an algorithm $I'$ which solves $\enum{A}$ with delay $O(p(n)q(k)+s)$ between the $k^\th$ and the $(k+1)^\th$ output solution, where $s$ is a bound on the size of a solution and $q(k) = (k+1)^{a+1}-k^{a+1}$. The algorithm $I'$ simulates $I$ together with a counter $c$ that is incremented at each step of $I$ and a counter $k$ which is initially set to $1$. Each time $I$ outputs a solution, we append it to a queue $\ell$ instead. Each time $c$ reaches $p(n)k^{a+1}$, the first solution of $\ell$ is output, removed and $k$ is incremented.

It is easy to see that during the execution of $I'$, $k-1$ always contains the number of solutions that have been output by $I'$ so far. Thus when $c$ reaches $p(n)k^{a+1}$, $I$ is guaranteed to have found $k$ solutions and $I'$ has only output $k-1$ of them, thus $\ell$ is not empty or it is the end of the execution of $I$. Moreover, the time elapsed between $I'$ outputs the $k^\th$ and the $(k+1)^\th$ solutions is the time needed to update the counters, plus the time needed to write a solution which is linear in $s$ plus  $(k+1)^{a+1}p(n)-k^{a+1}p(n) = q(k)p(n)$. Thus, the delay of $I'$ between the $k^\th$ and the $(k+1)^\th$ output solution is $O(p(n)q(k)+s)$. Since $\enum{A} \in \IncP_{a+1}$, we also have $\enum{A} \in \EnumP$, thus $s$ is polynomial in $n$. Moreover $q(k) = O(k^a)$. That is, $\enum{A} \in \UIncP_a$. 
\end{proof}

By choosing $a=0$ in Proposition~\ref{prop:UInca}, we directly get the interesting following result:
\begin{corollary}\label{prop:Inc1}
 $\IncP_1 = \DelayP$ and $\IncP = \UIncP$.
\end{corollary}

An inconvenience of Proposition~\ref{prop:UInca} is that the method used to go from our notion of incremental polynomial time to the usual notion of incremental time may blow up the memory. In practice, incremental delay is relevant if we also use only polynomial space. This naturally raises the question of understanding the relationship between $\IncP_{a+1}$ and $\UIncP_a$ when the space is required to be polynomial in the size of the input. 

As the more relevant classes in practice are $\DelayP$ and $\IncP_1$, we are concretely  interested in the following question: does every problem in $\IncP_1$ with polynomial space also have an algorithm in $\DelayP$ with polynomial space? Unfortunately, no classical assumptions in complexity theory seem to help for separating these classes nor were we able to prove the equality of both classes. The rest of this section is dedicated to particular $\IncP_1$ algorithms where the enumeration is sufficiently regular to be transformed into $\DelayP$ algorithm without blowing up the memory.

An algorithm $I$ is \emph{incremental linear} if there exists a polynomial $h$ such that on any instance of size $n$, it produces $k$ solutions in time bounded by $kh(n)$. We call $h$ the \emph{average delay} of $I$. By definition, a problem $\enum{A} \in \EnumP$ is in $\IncP_1$ if and only if there exists an incremental linear algorithm solving $\enum{A}$. 

Let $I$ be an incremental linear algorithm. Recall that $T(I,x,i)$ is number of steps made by $I$ before outputting the $i^{\text{th}}$ solution.
To make notations lighter, we will write $T(i)$ since $x$ and $I$ will be clear from the context. Consider a run of $I$ on the instance $x$, 
we will call $m_i$ an encoding of $i$, the memory of $I$ and its state at the time it outputs the $i^{\text{th}}$ solution. We say that the index $i$ is a \emph{$p$-gap} of $I$ if $T(i+1) - T(i) > p(|x|)$. If $I$ has no $p$ gaps for some polynomial $p$, it has polynomial delay $p$. We now show that when the number of large gaps is small, we can turn an incremental linear algorithm into a polynomial delay one, by computing shortcuts in advance.

\begin{proposition}
 Let $\enum{A} \in \IncP_1$ and $I$ be an incremental linear algorithm for $\enum{A}$ using polynomial space. 
 Assume there are two polynomials $p$ and $q$ such that for all instances $x$ of size $n$, there are at most $q(n)$ $p$-gaps
 in the run of $I$, then $\enum{A} \in \DelayP$.
\end{proposition}

 \begin{proof}
 Since $I$ is incremental linear it has a polynomial average delay that we denote by $h$.
 We run in parallel two copies of the algorithm $I$ that we call $I_1$ and $I_2$. When $I_1$ simulates one computation step of $I$, $I_2$ simulates $2h(n)$ computation steps of $I$. Moreover $I_2$ counts the number of consecutive steps without finding a new solution so that it detects
 $p$-gaps. When it detects such a gap, a pair $(i,m_{i+1})$ is stored where $i$ is the index of the last solution before the gap and $m_{i+1}$ is the description of the machine when it outputs the $(i+1)^{\text{th}}$ solution. Since there are at most $q(n)$ $p$-gaps and because $I$ uses polynomial space, the memory used by $I_2$ is polynomial. 
When $I_1$ outputs a solution of index $i$ and that $(i,m_{i+1})$ was stored by $I_2$, its state and memory is changed to $m_{i+1}$. 
Assume $I_2$ finds a gap at index $i$, then because $I$ is incremental linear, we have $T(i+1)- T(i) < (i+1)h(n)$.
Therefore $I_1$ at the same time has done at most $\frac{i+1}{2}$ computation steps and thus has not yet seen the $i^{\text{th}}$ solution, 
which proves that the algorithm works as described. In that way, $I_1$ will always generate solutions with delay less than $p(n)h(n)$ because $I_1$ has no $p$-gaps by construction, and each of its computation steps involves $h(n)$ computation steps of $I_2$.  
 \end{proof}

We can prove something more general, by requiring the existence of a large interval of solutions without $p$-gaps
 rather than bounding the number of gaps. It captures more cases, for instance an algorithm which outputs an exponential number of solutions at the beginning without gaps and then  has a superpolynomial number of gaps. The idea is to compensate for the gaps by using the dense parts of the enumeration.

\begin{proposition}
 Let $\enum{A} \in \IncP_1$ and $I$ be an incremental linear algorithm for $\enum{A}$ using polynomial space. 
 Assume there are two polynomials $p$ and $q$ such that for all $x$ of size $n$, and for all $k \leq |A(x)|$
 there exists  $a < b \leq k$ such that $b - a > \frac{k}{q(n)}$ and there are no $p$-gaps between the $a^\th$ and the $b^\th$ solution. Then $\enum{A} \in \DelayP$.
\end{proposition}
\begin{proof}
We let $h$ be the average delay of $I$. We fix $x$ of length $n$ and describe a process that enumerates $A$ with delay at most $2q(n)h(n) \cdot (q(n)h(n)+p(n))$ and polynomial space on input $x$. Our algorithm runs two processes in parallel: $\En$, the enumerator and $\Ex$, the explorer. Both processes simulate $I$ on input $x$ but at a different speed that we will fix later in the proof. $\En$ is the only one outputting solutions. We call a solution {\em fresh} if it has not yet been enumerated by $\En$.

$\Ex$ simulates $I$ and discovers the boundaries of the largest interval without $p$-gaps containing only fresh solutions that we call the \emph{stock}. More precisely, it stores two machine states: $m_a$ and $m_b$ where $a$ and $b$ correspond to indices of fresh solutions such that there are no $p$-gaps between $a$ and $b$ and it is the largest such interval. Intuitively, the stock contains the fresh solutions that will make up for $p$-gaps in the enumeration of $I$.

$\En$ can work in two different modes. If $\En$ is in simple mode, then it only simulates $I$ on input $x$ and outputs a solution whenever $I$ outputs one and counts the number of steps between two solutions. When it detects a $p$-gap, $\En$ switches to filling mode. In filling mode, $\En$ starts by copying $m_a$ into a new variable $s$ and $m_b$ into a new variable $t$. It then runs two simulations of $I$: the first one is the continuation of the simulation that was done in simple mode. The second one, which we call the {\em filling simulation} is a simulation of $I$ starting in state $m_a$. $\En$ simulates $h(n)q(n)$ steps of the first enumeration and then as many steps as necessary to find the next solution of the filling simulation. Since the stock does not contain $p$-gaps by definition, we know that $\En$ outputs solutions with delay at most $h(n)q(n)+p(n)$. To avoid enumerating the same solution twice, whenever the first simulation reaches the state stored in $s$, we stop the first simulation and $\En$ switches again in simple mode using the filling simulation as starting point. 

We claim that the first simulation will always reach state $s$ before the filling simulation reach the end of the stock. Indeed, assume that the filling simulation has reached the end of the stock and outputs the $b^\th$ solution. By definition, the stock is the largest interval without $p$-gaps before this solution and it is of size at least $b/q(n)$ by assumption. Thus, the first simulation has simulated at least $(b/q(n))h(n)q(n) = b \cdot h(n)$ steps of $I$ in parallel. Thus, by definition of $h$, the $b$ first solutions have been found by the first simulation. It must have reached state $s$ before the filling simulation reaches the end of the interval.

Using this strategy, it is readily verified that if $\En$ has always a sufficiently large stock at hand, it enumerates $A(x)$ entirely with delay at most $h(n)q(n)+p(n)$. 

We now choose the speed of $\Ex$ in order to guarantee that the stock is always sufficiently large: each time $\En$ simulates one step of $I$, $\Ex$ simulates $2h(n)q(n)$ steps of $I$.

There is only one situation that could go wrong: the enumerator can reach state $m_b$, which is followed by a $p$-gap while the explorer has not found a new stock yet. We claim that having chosen the speed as we did, we are guaranteed that it never happens. Indeed, if $\En$ reaches $m_b$, then it has already output $b$ solutions. Thus, $\Ex$ has already simulated at least $b \cdot 2h(n)q(n)$ steps of $I$. By definition of $h$, $\Ex$ has already found $2b \cdot q(n)$ solutions and then, it has found an interval without $p$-gaps of size $(2b \cdot q(n))/q(n) = 2b$ which is necessarily ahead of the simulation of $\En$.

The property of the last paragraph is only true if the simulation of $I$ by $\Ex$ has not stopped before $b \cdot 2h(n)q(n)$ steps. To deal with this case, as soon as $\Ex$ has stopped, $\En$ enters in filling mode if it was not
in this mode and does a third simulation in parallel of $I$ beginning at state $m_b$. This takes care of the solutions after the last stock.
\end{proof}
%

Note that in both proofs the polynomial delay we obtain is worse than the average delay of the incremental algorithm but 
the total time is the same.
Also we do not use all properties of an algorithm in $\IncP_1$ but only the fact that the predicate is polynomially balanced.
All known algorithms which are both incremental and in polynomial space are in fact polynomial delay algorithms 
with a bounded number of repetitions and a polynomial time algorithm to decide whether it is the first time a solution is produced~\cite{phdstrozecki}. It seems that if we can turn such an algorithm to one in polynomial delay, we would have solved the general problem.

\begin{openproblem}
Prove or disprove that $\IncP_1$ with polynomial space is equal to $\DelayP$ with polynomial space.
\end{openproblem}

Here we tried to improve the regularity of an algorithm without losing too much memory. 
The opposite question is also natural: is it possible to trade regularity and total time for space in enumeration.
In particular can we improve the memory used by an enumeration algorithm if we are relaxing the constraints on the delay.

\begin{openproblem}
 Can we turn a polynomial delay algorithm using an exponential space memory into an output polynomial 
or even an incremental polynomial algorithm with polynomial memory ? 
 \end{openproblem}

\section{Strict hierarchy in incremental time problems}
\label{sec:hier}

We prove strict hierarchies for $\IncP_{a}^F$ unconditionally and for $\IncP_a$ modulo the {\em Exponential Time Hypothesis} ($\ETH)$. Since $\DelayP \subseteq \IncP_1$ it implies that $\DelayP \neq \IncP$ modulo $\ETH$.

\begin{proposition}
 $\IncP^F_a \subsetneq \IncP^F_b$ when $1\leq a < b$.
\end{proposition}

\begin{proof}
By the time hierarchy theorem~\cite{hartmanis1965computational}, there exists a language $L$ which can be decided in time $O(2^{nb})$ but not in time $O(2^{na})$. Let $n = |x|$. We build a predicate $A(x,y)$ which is true if and only if either $y$ is a positive integer written in binary with $y < 2^{n}$ or $y = \sharp 0$ when $x \notin L$ or $y = \sharp 1$ when $x \in L$. We have an algorithm to solve $\enum{A}$: first enumerate the $2^n$ trivial solutions then run the $O(2^{nb})$ algorithm which solves $A$ to compute the last solution. This algorithm is in $\IncP_b^F$, since finding the $2^n$ first solutions can be done in $\IncP_1$ and the last one can be found in time $O((2^n)^b)$.
Assume there is an $\IncP_a^F$ algorithm to solve $\enum{A}$ with a precomputation step bounded by the polynomial $p(n)$.
By running this enumeration algorithm for a time $O(p(n) + 2^{na}) = O(2^{na})$ we are guaranteed to find all solutions.
Therefore one finds either the solution $\sharp 0$ or $\sharp 1$ in time $O(2^{na})$ which is a contradiction therefore $\enum{A} \notin \IncP_a^F$.
\end{proof}

This proof can easily be adapted to prove an unconditional hierarchy inside $\OutputP^F$ and $\DelayP^F$.
In the case of $\DelayP^F$, one must use a padding and a complexity for $L$ of $n^{\log(n)}$ to dominate the precomputation step while satisfying the hypothesis of the time hierarchy theorem.

To prove the existence of a strict hierarchy in $\IncP$, we need to assume some complexity hypothesis since 
$\P = \NP$ implies $\IncP = \IncP_1$ by the same argument as in the proof of Prop.~\ref{prop:output}.
Moreover, the hypothesis must be strong enough to replace the time hierarchy argument.

The Exponential Time Hypothesis states that there exists $\epsilon > 0$ such that there is no algorithm for $\kSAT{3}$ in time $\tilde{O}(2^{\epsilon n})$ where $n$ is the number of variables of the formula and $\tilde{O}$ means that we have a factor of $n^{O(1)}$. The {\em Strong Exponential Time Hypothesis} ($\SETH$) states that for every $\epsilon < 1$, there is no algorithm solving $\SAT$ in time $\tilde{O}(2^{\epsilon n})$.

We show that if $\ETH$ holds, then $\IncP_a \subsetneq \IncP_b$ for all $a < b$. For $t \leq 1$, let $R_t$ be the following predicate: given a CNF $\phi$ with $n$ variables, $R_{t}(\phi)$ contains:

\begin{itemize}
\item the integers from $1$ to $2^{nt}-1$
\item the satisfying assignments of $\phi$ duplicated $2^n$ times each, that is $SAT(\phi) \times [2^n]$.
\end{itemize}

We let $\Pad{t}$ be the enumeration problem associated to $R_t$, that is $\Pad{t} = \enum{R_t}$. The intuition behind $\Pad{t}$ is the following. Imagine that $t = b^{-1}$ for some $b \in \N$. By adding sufficiently many dummy solutions to the satisfying assignments of a CNF-formula $\phi$, we can first enumerate them quickly and then have sufficient time to bruteforce $SAT(\phi)$ in $\IncP_b$ before outputting the next solution. This shows that $\Pad{b^{-1}} \in \IncP_b$.
Now, if there exists $a < b$ such that $\IncP_a = \IncP_b$, we would have a way to find a solution of $\phi$ in time $\tilde{O}(2^{{a \over b}n})$ which already violate $\SETH$. To show that we also violates $\ETH$ we repeat this trick but we do not bruteforce $\SAT(\phi)$ anymore. We can do better by using this $\tilde{O}(2^{{a \over b}n})$ algorithm for $\SAT$ and we can gain a bit more on the constant in the exponent. We show that by repeating this trick, we can make the constant as small as we want. We formalize this idea:
\begin{lemma}
\label{lem:risin}
  Let $d<1$. If we have an $\tilde{O}(2^{dn})$ algorithm for $\SAT$, then for all $b \in \N$, $\Pad{d \over b}$ is in $\IncP_b$.
\end{lemma}
\begin{proof}
  We enumerate the integers from $1$ to $2^{dn \over b}-1$ and then call the algorithm to find a satisfying assignment of $\phi$. We have enough time to run this algorithm since the time allowed before the next anser is $\tilde{O}\Big(\big({2^{dn \over b}}\big)^b\Big) = \tilde{O}(2^{dn})$. If the formula is not satisfiable, then we stop the enumeration. Otherwise, we enumerate all copies of the discovered solution. We have then enough time to bruteforce the other solutions.
\end{proof}

\begin{lemma}
\label{lem:algosat}
  If $\Pad{t}$ is in $\IncP_a$, then there exists an $\tilde{O}(2^{nta})$ algorithm for $\SAT$.
\end{lemma}
\begin{proof}
  Since $\Pad{t}$ is in $\IncP_a$, we have an algorithm for $\Pad{t}$ that outputs $m$ elements of $R_t(\phi)$ in time $O(m^a |\phi|^c)$ for a constant $c$. We can then output $2^{nt}$ elements of $R_t(\phi)$ in time $O(2^{nta}|\phi|^c) = \tilde{O}(2^{nta})$. If the enumeration stops before having output $2^{nt}$ solutions, then the formula is not satisfiable. Otherwise,  we have necessarily enumerated at least one satisfying assignment of $\phi$ which gives the algorithm.
\end{proof}

\begin{lemma}
\label{lem:amplification}
If $\IncP_a = \IncP_b$, then for all $i \in \N$, $\Pad{a^i \over b^{i+1}}$ is in $\IncP_a$.
\end{lemma}
\begin{proof}
  The proof is by induction on $i$. For $i=0$, by Lemma~\ref{lem:risin}, $\Pad{1 \over b}$ is in $\IncP_b$ since we have an $\tilde{O}(2^n)$ bruteforce algorithm for $\SAT$. Thus, if $\IncP_b = \IncP_a$, $\Pad{1 \over b}$ is in $\IncP_a$ too.

  Now assume that $\Pad{a^i \over b^{i+1}}$ is $\IncP_a$. By Lemma~\ref{lem:algosat}, we have an $\tilde{O}(2^{dn})$ algorithm for $d = {a^{i+1} \over b^{i+1}}$. Thus, by Lemma~\ref{lem:risin}, $\Pad{d \over b} = \Pad{a^{i+1} \over b^{i+2}}$ is in $\IncP_b = \IncP_a$.
\end{proof}

\begin{theorem}
\label{thm:separation}
  If $\ETH$ holds, then $\IncP_a \subsetneq \IncP_b$ for all $a<b$.
\end{theorem}
\begin{proof}
  If there exists $a < b$ such that $\IncP_a = \IncP_b$, then by Lemma~\ref{lem:amplification}, for all $i$, $\Pad{a^i \over b^{i+1}}$ is in $\IncP_a$. Thus by Lemma~\ref{lem:algosat}, we have an $\tilde{O}(2^{d_in})$ algorithm for $\SAT$ and then for $\kSAT{3}$ in particular, where $d_i = \big ( {a \over b} \big )^{i}$. Since $\lim_{i \rightarrow \infty} d_i = 0$, this contradicts $\ETH$.
\end{proof}

Observe that by Proposition~\ref{prop:UInca} and Theorem~\ref{thm:separation}, we also have that if $\ETH$ holds, then we also have a strict hierarchy inside $\UIncP$. 

In the previous proofs, we did not really used $\SAT$. We needed an $\NP$ problem, 
with a set of easy to enumerate potential solutions of size $2^n$ that cannot be solved in time $2^{o(n)}$.
For instance we could use $\mathsf{CIRCUIT\textsf{-}SAT}$ which is the problem of finding a satisfying assignment to a Boolean circuit. We can thus prove our result by assuming a weaker version of $\ETH$ as it is done in~\cite{AbboudHWW16}. It would be nice to further weaken the hypothesis, but it seems hard to rely only on a classical complexity hypothesis such as $\P \neq \NP$.
%
%
The other way we could improve this result, is to prove a lower bound for a natural enumeration problem instead of $\Pad{t}$.

\begin{openproblem}
  Prove that enumerating the minimal transversals of an hypergraph cannot be done in $\IncP_1$ if $\ETH$ hods.
\end{openproblem}

It is also natural to try to obtain the same hierarchy for $\DelayP$. However, the difference in total time 
between two algorithms with different polynomial delays is very small and the proof for the separation of the incremental hierarchy does not seem to carry on.

\begin{openproblem}
  Prove there is a strict hierarchy inside $\DelayP$ assuming $\SETH$ or even stronger hypotheses.
\end{openproblem}

\section{From Uniform Generator to efficient randomized enumeration}\label{sec:uniform}

In this section, we explore the relationship between efficient enumeration and random generation of combinatorial structures or sampling. The link between sampling and counting combinatorial structures has already been studied. For instance, Markov Chain Monte Carlo algorithms can be used to compute an approximate number of objects~\cite{jerrum2003counting} or in the other direction, generating functions encoding the number of objects of each size can be used to obtain Boltzman samplers~\cite{duchon2004boltzmann}. 

In her thesis~\cite{Goldberg91} (Section 2.1.2), Leslie Goldberg proved several results relating the existence of a good sampling algorithm for a set $S$ with the existence of an efficient algorithm to enumerate $S$. In this section, we review these results and improve the runtime of the underlying algorithms. Moreover, we show that if we allow repetitions during the enumeration, we can design algorithms  using only polynomial space. This complements a result by Goldberg showing a space-time trade-off if we do not relax the notion of enumeration.


\begin{definition}
 Let $\enum{A} \in \EnumP$. A {\em polytime uniform generator} for $A$ is a randomized RAM machine $M$ which outputs an element $y$ of $A(x)$ in time polynomial in $|x|$ such that the probability over every possible run of $M$ on input $x$ that $M$ outputs $y$ is ${|A(x)|}^{-1}$.
%
\end{definition} 

We now define a randomized version of $\IncP$, which has first been introduced in~\cite{phdstrozecki,strozecki2013enumerating}
to capture random polynomial interpolation algorithms.

\begin{definition}
\label{def:incp1rand}
A problem $\enum{A}$ is in randomized $\IncP_k$ if $\enum{A} \in \EnumP$ and there exists constants $a,b,c \in \N$ and a randomized RAM machine $M$ such that for every $x \in \{0,1\}^*$ and $\epsilon \in \mathbb{Q}_+$, the probability that $M$, on input $x$ and $\epsilon$, enumerates $A(x)$ in incremental time $cm^k n^a\epsilon^{-b}$ is greater than $1-\epsilon$.
\end{definition}

Definition~\ref{def:incp1rand} can be understood as follows, on input $x \in \{0,1\}^*$ and $\epsilon \in \mathbb{Q}_+$, the probability of the following fact is at least $1-\epsilon$: for every $t \leq |A(x)|$, $M$ has enumerated $t$ distinct elements of $A(x)$ after $ct^kn^a\epsilon^{-b}$ steps and stops in time less than $|A(x)|^kn^a\epsilon^{-b}$.

\begin{theorem}
\label{thm:gentoincp1}
 If $\enum{A} \in \EnumP$ has a polytime uniform generator, then $\enum{A}$ is in randomized $\IncP_1$.
\end{theorem}
\begin{proof}
\begin{algorithm}
 \KwData{$x \in \{0,1\}, \epsilon \in \mathbb{Q}_+$}
 \Begin{
   $E \leftarrow \emptyset$; 
   $r \leftarrow 0$\;
   $K \leftarrow 2 \cdot \big(p(|x|)-\log(\epsilon/2)\big)$\;
   \While{$r \leq K \cdot |E|$}{
     Draw $e \in A(x)$ uniformly and $r \leftarrow r+1$\;
     \If{$e \notin E$}
     {
       Output $e$ and $E \leftarrow E \cup \{e\}$; 
     }
   }
 }
  \caption{An algorithm to enumerate $\enum{A}$ in randomized $\IncP_1$, where every element of $A(x)$ is of size at most $p(|x|)$.}
  \label{alg:generatorenum}
\end{algorithm}

Algorithm~\ref{alg:generatorenum} shows how to use a generator for $A$ to enumerate its solutions in randomized $\IncP_1$. The idea is the most simple: we keep drawing elements of $A(x)$ uniformly by using the generator. If the drawn element has not already been enumerated, then we output it and remember it in a set $E$. We keep track of the total number of draws in the variable $r$. If this variable reaches a value that is much higher than the number of distinct elements found at this point, we stop the algorithm. We claim that Algorithm~\ref{alg:generatorenum} is in randomized $\IncP_1$, the analysis is similar to the classical coupon collector theorem~\cite{erdos1961classical}.

We let $p$ be a polynomial such that for every $x \in \{0,1\}^*$, the size of elements of $A(x)$ is a most $p(|x|)$. Such a polynomial exists since $\enum{A} \in \EnumP$. Observe that all operations can be done in polynomial time in $|x|$ since we can encode $E$ -- the set of elements that have already been enumerated -- by using a datastructure such as a trie for which adding and searching for an element may be done in time $O(p(|x|)$, the size of the element.

Moreover, observe that if the algorithm is still running after $tK$ executions of the while loop, then we have $|E| \geq t$, thus we have enumerated more than $t$ elements of $A(x)$. Since each loop takes a time polynomial in $|x|$ and $K$ is polynomial in $|x|$ and $\epsilon$, we have that if the algorithm still runs after a time $t \cdot \poly(|x|)$, then the run is similar to a run in $\IncP_1$.

Hence, to show that $\enum{A}$ is in randomized $\IncP_1$, it only remains to prove that the probability that Algorithm~\ref{alg:generatorenum} stops before having enumerated $A(x)$ completely is smaller than $\epsilon$. The main difficulty is to decide when to stop. It cannot be done deterministically since we do not know $|A(x)|$ {\em a priori}. Algorithm~\ref{alg:generatorenum} stops when the total number of draws $r$ is larger than $K \cdot |E|$, where $E$ is the set of already enumerated elements of $A(x)$. In the rest of the proof, we prove that with $K =  2 \cdot \big(p(|x|)-\log(\epsilon/2)\big)$, Algorithm~\ref{alg:generatorenum} stops after having enumerated $A(x)$ completely with probability greater than $1-\epsilon$.

In the following, we fix $x \in \{0,1\}^*$ and $\epsilon \in \mathbb{Q}_+$. We denote by $s = |A(x)|$ the size of $A(x)$. Remember that we have $s \leq 2^{p(|x|)}$. We denote by $T$ the random variable whose value is the number of distinct elements of $A(x)$ that have been enumerated when the algorithm stops. Our goal is to show that $\pr(T < s) \leq \epsilon$. 

We start by showing that $\pr(T \leq s/2) \leq \epsilon/2$. Let $t \leq s/2$. We bound the probability that $T = t$. If the algorithm stops after having found $t$ solutions, we know that it has found $t$ solutions in less than $1+(t-1)K$ draws, otherwise, if the algorithm had found less than $t$ solutions after $1+(t-1)K$ draws then the while loop would have finished. After that, it keeps on drawing already enumerated solutions until it has done $1+tK$ draws and stops. Thus, it does at least $K$ draws without finding new solutions. Since $t \leq s/2$, the probability of drawing a solution that was already found is at most $1/2$. Thus for all $t \leq s/2$,
\[ \pr(T=t) \leq 2^{-K} \leq 2^{-p(|x|)}(\epsilon/2) \]
since $K \geq -\log(\epsilon/2)$. Now, applying the union bound yields:
\[\pr(T \leq s/2) \leq \sum_{t=1}^{s/2} \pr(T=t) \leq (s/2)2^{-p(|x|)}(\epsilon/2) \leq \epsilon/2 \]
since $s \leq 2^{p(|x|)}$.

Now, we show that $\pr(s/2 < T < s) \leq \epsilon/2$. Assume that $T > s/2$. Then, after $K \cdot (s/2)$ draws, the algorithm has not stopped. Thus the probability that $s/2 < T < s$ is smaller than the probability that, after $r = K \cdot (s/2)$ draws, we have not found every element of $A(x)$. Given an element $y \in A(x)$, the probability that $y$ is not drawn after $r$ draws is $(1-1/s)^r$. Thus, the probability that after $r = K \cdot (s/2)$ draws, we have not found every element of $A(x)$ is at most 
\[ 
\begin{aligned}
s \cdot (1-1/s)^r  & \leq s \cdot 2^{-r/s} &\\
& \leq s \cdot 2^{\log(\epsilon/2)-p(|x|)} & \text{since } r = K \cdot (s/2)\\
& \leq \epsilon/2 \cdot s2^{-p(|x|)} &\\
& \leq \epsilon/2 & \text{since } s \leq 2^{p(|x|)}
\end{aligned}
\]

In the end, $\pr(T < s) \leq \pr(T \leq s/2) + \pr(s/2 < T < s/2) \leq \epsilon$. We observe that the running time of Algorithm~\ref{alg:generatorenum} is actually polynomial in $\log(\epsilon^{-1})$ which is a much better bound than the one of Definition~\ref{def:incp1rand} since $\log(\epsilon^{-1})$ is polynomial in the size of the encoding of $\epsilon$ for $\epsilon < 1$.
\end{proof}

Applying the same technique as Prop.~\ref{prop:Inc1}, we can turn Algorithm~\ref{alg:generatorenum} into a randomized $\DelayP$ algorithm by amortizing the generation of solutions.

\begin{corollary}
 If $\enum{A} \in \EnumP$ has a polytime uniform generator, then $\enum{A}$ is in randomized $\DelayP$.
\end{corollary}

In~\cite{Goldberg91}, Goldberg uses generators that are not necessarily uniform and may be biased by a factor $b$. We can easily modify Algorithm~\ref{alg:generatorenum} to make it work with a biased generator.
\begin{definition}
 Let $\enum{A} \in \EnumP$ and $b$ a polynomial. A {\em polytime $b$-biased generator} for $A$ is a randomized RAM machine $M$ which outputs an element $y$ of $A(x)$ in time polynomial in $|x|$ such that the probability over every possible run of $M$ on input $x$ that $M$ outputs $y$ is at least ${|A(x)|b(x)}^{-1}$.
\end{definition} 

\begin{theorem}
\label{thm:gentoincp1biased}
 If $\enum{A} \in \EnumP$ has a polytime $b$-biased generator, then $\enum{A}$ is in randomized $\IncP_1$.
\end{theorem}
\begin{proof}[Proof (sketch).]
It is sufficient to replace $K \leftarrow 2 \cdot \big(p(|x|)-\log(\epsilon/2)\big)$ in Algorithm~\ref{alg:generatorenum} by $K \leftarrow 2 \cdot b(|x|) \cdot \big(p(|x|)-\log(\epsilon/2)\big)$. The proof follows then exactly the proof of Theorem~\ref{thm:gentoincp1}. 
\end{proof}

Theorem~\ref{thm:gentoincp1biased} is an improved version of Theorem 2, Section 2.1.2 in~\cite{Goldberg91}. It is not hard to see that in our algorithm, the average delay between two solutions is $O(p(|x|)g(|x|)b(|x|))$ where $g$ is the runtime of the generator. The average delay of Goldberg's algorithm is, with our notations, $O(p(|x|)^3g(|x|)b(|x|))$.


\subparagraph{Polynomial space algorithm.}  
The main default of Algorithm~\ref{alg:generatorenum} is that it stores all solutions enumerated and therefore needs a space which may be exponential. It seems necessary to encode the subset of already generated solutions and these subsets are in doubly exponential number and thus cannot be encoded in polynomial space. Therefore the enumeration algorithm needs time to rule out a large number of possible subsets of generated solutions. This idea has been made precise by Goldberg (Theorem $3$, p.$33$~\cite{Goldberg91}): the product of the delay and the space is lower bounded by the number of solutions to output up to a polynomial factor. On the other hand it is easy to build an enumeration algorithm with such space and delay, by generating solutions 
by blocks in lexicographic order (Theorem $5$, p.$42$~\cite{Goldberg91}). The proof of the lower bound uses the fact that the enumeration algorithm can only output solutions which are given by calls to the generator. A set of possible initial sequences of output elements in the enumeration is built so that its cardinality is bounded by an exponential in the space used and that the enumeration produces one of these sequences with high probability. Then if the delay is to small, with high probability the calls to the generator has not produced any of those special sequences which ends the proof.
 
However, if we allow \emph{unbounded repetitions} of solutions in the enumeration algorithm we can devise an incremental polynomial algorithm with polynomial space. The main difficulty is again to decide when to stop so that no solution is forgotten with high probability. The method used in Algorithm~\ref{alg:generatorenum} does not work in this case since we cannot maintain the number of distinct solutions that have been output so far. However, there exists data structures which allow to approximate the cardinal of a dynamic set using only a logarithmic number of bits in the size of the set~\cite{flajolet1985probabilistic,kane2010optimal}.
The idea is to apply a hash function to each element seen and to remember an aggregated information on the bits of the hashed elements. Algorithm~\ref{alg:generatorenumspace} shows how we can exploit such datastructures to design a randomized incremental algorithm from a uniform generator. Unlike Algorithm~\ref{alg:generatorenum}, Algorithm~\ref{alg:generatorenumspace} cannot be turned into a polynomial delay algorithm since it would require exponential space and our improvement would then be useless.
 
\begin{algorithm}
 \KwData{$x \in \{0,1\}, \epsilon \in \mathbb{Q}_+$}
 \Begin{
   Initialize $E$;
   $r \leftarrow 0$\;
   $K \leftarrow 4 \cdot \big(p(|x|)-\log(\epsilon/4)\big)$\;
   \While{$r \leq K \cdot \mathsf{estimate}(E)$}{
     Draw $e \in A(x)$ uniformly and
     $r \leftarrow r+1$\;
     Output $e$ and 
     $\mathsf{update}(E,e)$\;
     }
   }
 
  \caption{An algorithm in randomized $\IncP_1$ with polynomial space such that every element of $A(x)$ is of size at most $p(|x|)$.}
  \label{alg:generatorenumspace}
\end{algorithm}
\begin{proposition}
 If $\enum{A} \in \EnumP$ has a polytime uniform generator, then there is an enumeration algorithm 
 in randomized $\IncP_1$ \emph{with repetitions} and \emph{polynomial space}.
\end{proposition}

\begin{proof}
 The procedure $\mathsf{update}$ in Algorithm~\ref{alg:generatorenumspace} maintains a data structure which allows $\mathsf{estimate}$ to output an approximation of $|E|$.
 If we use the results of~\cite{kane2010optimal}, we can get a $2$-approximation of $|E|$ during \emph{all} the algorithm with probability $1-\epsilon/2$. The data structure uses a space $\log(|E|)\log(\epsilon^{-1})$. The process $\mathsf{update}(E,e)$ does $O(\log(\epsilon^{-1}))$ arithmetic operations and $\mathsf{estimate}$ does $O(1)$ arithmetic operations.
 The arithmetic operations are over solutions seen as integers which are of size polynomial in $n$. The analysis of the delay is the same as before, but to the cost of generating a solution, we add the cost of computing $\mathsf{update}(E,e)$ and $\mathsf{estimate}$ which are also polynomial in $n$. 

 The analysis of the correctness of the algorithm is the same, except that now $v$ is a $2$-approximation of $|E|$ with probability
 $1-\frac{\epsilon}{2}$. We have adapted the value of $K$ such that with probability $\epsilon/2$ the algorithm will not stop before generating all solutions even if 
 $|E|$ is approximated by $|E|/2$. Therefore the probability to wrongly evaluate $|E|$ plus 
 the probability that the algorithm stops too early is less than $\epsilon$.
 
\end{proof}

The method we just described here can be relevant, when we have an enumeration algorithm using the supergraph method: a connected graph whose vertices are all the solutions is defined in such a way that the edges incident to a vertex can be enumerated with polynomial delay. The enumeration algorithm does a traversal of this graph which requires to store all generated solutions to navigate the graph.
The memory used can thus be exponential. On the other hand doing a random walk over the graph of solutions often yields a polynomial time uniform generator. If it is the case using Algorithm~\ref{alg:generatorenumspace} we get a randomized polynomial delay  algorithm using polynomial space only.

The more classical way to avoid exponential memory is Lawler's method~\cite{LawlerLK80} or reverse search, that is defining an implicit spanning tree in the graph which can be navigated with polynomial memory. This method is not always relevant since it based on solving a search problem which may be $\NP$-hard.  One could also traverse the graph of solutions using only a logarithmic space in the numbers of solutions using a universal sequence~\cite{reingold2008undirected} but this method gives no guarantee on the delay and has a huge slowdown in practice.

\section*{Acknowledgement}

We are thankfull to Arnaud Durand for numerous conversations about enumeration complexity 
and for having introduced the subject to us. This work was partially supported by the French Agence Nationale de la Recherche, AGGREG project reference ANR-14-CE25-0017-01 and by the ESPRC grant EP/LO20408/1.

\bibliography{biblio}
\bibliographystyle{plain}
\end{document}